\documentclass{article}
\usepackage{ijcai21}

\usepackage{times}
\usepackage{soul}
\usepackage{url}
\usepackage[hidelinks]{hyperref}
\usepackage[utf8]{inputenc}
\usepackage[small]{caption}
\usepackage{graphicx}
\usepackage{amsmath}
\usepackage{amsthm}
\usepackage{booktabs}
\usepackage{algorithm}
\usepackage{amssymb,amsfonts}
\usepackage{textcomp}
\usepackage{xcolor}

\usepackage{comment}
\usepackage{url}
\usepackage[T1]{fontenc}
\usepackage[utf8]{inputenc}
\usepackage{blindtext}
\usepackage{pdfpages}

\usepackage{here}
\usepackage{xcolor}
\usepackage{fancyhdr}
\excludecomment{hdn}
\usepackage[title,titletoc]{appendix}

\usepackage{ascmac}
\usepackage{bm,enumerate,amsmath,amssymb,amsthm}
\newtheoremstyle{assumption}
  {\topsep}
  {\topsep}
  {}
  {1em}
  {\itshape}
  {}
  {.5em}
  {\thmname{#1}\thmnumber{ #2}\thmnote{ (#3)}}
\theoremstyle{assumption}

\newtheorem{lemma}{Lemma}
\newtheorem{theorem}{Theorem}

\newtheorem{minortheorem}{Proposition}

\usepackage[noend]{algpseudocode}
\algnewcommand\Input{\item[{\textbf{Input:}}]}
\algnewcommand\Output{\item[{\textbf{Output:}}]}
\algrenewcommand\algorithmicdo{}
\algnewcommand\algorithmicto{\textbf{to}}
\algnewcommand\algorithmicbreak{\textbf{break}}
\algdef{SE}[FOR]{ForTo}{EndFor}[2]{\algorithmicfor\ #1\ \algorithmicto\ #2\ \algorithmicdo}{\algorithmicend\ \algorithmicfor}%
\algtext*{EndFor}
\let\oldReturn\Return
\renewcommand{\Return}{\State\oldReturn}

\def \vect#1{\mbox{\boldmath $#1$}}
\newcommand{\mc}{\mathcal}

\newcommand{\argmin}{\mathop{\rm argmin~}\limits}

\newcommand{\newcolor}{black}

\newcommand{\modaaai}[1]{\textcolor{black}{#1}}
\newcommand{\modified}[1]{\textcolor{black}{#1}}
\newcommand{\extref}[1]{Equation #1 in the paper}

\def\arxiv{1}
\newcommand{\cameracolor}{black}

\if\arxiv0
\newcommand{\camera}[1]{\textcolor{\cameracolor}{#1}}
\else
\newcommand{\camera}[1]{\textcolor{\cameracolor}{#1}}
\fi
\newcommand{\urlblue}[1]{\textcolor{blue}{#1}}

\title{A Polynomial-time, Truthful, Individually Rational and Budget Balanced Ridesharing Mechanism}

\iffalse
\author{
    paper 3927
}
\else
\author{
Tatsuya Iwase$^1$
\and
Sebastian Stein$^2$
\and
Enrico H. Gerding$^2$
\affiliations
$^1$Toyota Motor Europe NV/SA, Zaventem, Belgium\\
$^2$University of Southampton, Southampton, United Kingdom\\
\emails
tiwase@mosk.tytlabs.co.jp,
\{ss2, eg\}@ecs.soton.ac.uk
}
\fi

\begin{document}

\maketitle

\if\arxiv1
\thispagestyle{fancy}
\pagestyle{fancy}
\fancyhf{}
\renewcommand{\headrulewidth}{0pt}
\lfoot[Copyright International
Joint Conferences on Artificial Intelligence (\href{https://ijcai.org}{IJCAI}), 2021. All rights
reserved.]{Copyright International
Joint Conferences on Artificial Intelligence (\href{https://ijcai.org}{\urlblue{IJCAI}}), 2021. All rights
reserved.}

\fi

\begin{abstract}
Ridesharing has great potential to improve transportation efficiency while reducing congestion and pollution. To realize this potential, mechanisms are needed that allocate vehicles optimally and provide the right incentives to riders. However, many existing approaches consider restricted settings (e.g., 
only one rider per vehicle
or a common origin for all riders). Moreover, 
naive applications of standard approaches, such as the Vickrey-Clarke-Groves or greedy mechanisms, cannot achieve a polynomial-time, truthful, individually rational and budget balanced mechanism. To address this, we formulate a general ridesharing problem and apply mechanism design to develop a novel mechanism which satisfies all four properties and whose social cost is within 8.6\% of the optimal on average.
\if\arxiv0
\renewcommand{\thefootnote}{}
\footnote[0]{\camera{An extended version including technical appendices is available at \href{https://arxiv.org/abs/2105.11292}{https://arxiv.org/abs/2105.11292}.}}
\renewcommand{\thefootnote}{\arabic{footnote}}
\fi
\end{abstract}

\setcounter{section}{0}

\section{Introduction}
\label{sec:intro}

Ridesharing can meet more traffic demand with fewer vehicles than existing taxi services or privately-owned vehicles. At the same time, it can reduce traffic congestion, energy consumption and pollution. Moreover, with the introduction of autonomous vehicles, unused vehicles become active transportation resources, and ridesharing can further increase utilization of such vehicles
\cite{Spieser2014}. 

Broadly, there are two research streams addressing ridesharing problems: \modified{the first }one focuses on algorithms for assigning vehicles to riders while solving the routing problem in an optimal way \cite{Laporte2007,Hasan2018,Chen2019}.
However, none of these studies consider incentives and whether it is in the best interest of riders to follow the proposed solution and to report their preferences truthfully. Especially when the origins and  destinations of the riders are different, they are sometimes required to detour to share the vehicle with others. Hence, there may be an incentive to misreport time constraints and preferences to prevent detours.

\color{\newcolor}
To address this, the second research stream focuses on mechanism design approaches, which incentivize truthful reporting and  avoid strategic manipulation. In addition to this property, which is called incentive compatibility (IC), other desirable properties include (weak) budget balance (BB, i.e., the system does not require an outside subsidy),  individual rationality (IR, i.e., an individual is never worse off participating in the mechanism) and scalability (i.e., the allocation and any payments can be computed in polynomial time). 
An early approach uses a second price auction \cite{Kleiner}, but their work is limited to single riders per vehicle and does not allow  cooperation between drivers. \cite{Crandall2001}  propose an online posted-price mechanism but they only consider autonomous vehicles (vehicles can move without any riders allocated) and their approach is not scalable. Furthermore, \cite{Cheng2014} propose a mechanism that satisfies IC, BB and IR, but they assume all riders depart from the same origin and so their method does not coordinate the pickup of riders at different locations. \cite{Zhao} propose a Vickrey–Clarke–Groves (VCG) mechanism, but their approach is not scalable and BB is only satisfied when there is no detour. More recent work on mechanism design, e.g., \cite{Rheingans-yoo2019,LuAlicePeterI.Frazier2018,MaHongyaoFeiFang2019}, only focuses on drivers and does not consider rider incentives. Also, none of the above consider settings where riders can switch vehicles which can enable more efficient solutions, especially in settings with limited vehicles as in suburban areas. 
\color{black}

In short, to date there exist no scalable mechanisms for general ridesharing settings that satisfy IC, IR and BB. To address this shortcoming, we make the following novel contributions:
\begin{itemize}
    \item We formulate a more general ridesharing problem compared to existing approaches. In particular, this is the first model that \modified{considers both autonomous and non-autonomous vehicles and where riders can switch between vehicles if it is beneficial to do so.} 
    \item \modified{We show that approaches based on VCG mechanisms or naive greedy (where riders are allocated in order of their marginal cost) do not satisfy the necessary properties. Specifically, the standard VCG approach does not satisfy BB and cannot be computed in polynomial-time. Also, naive greedy violates a property called monotonicity, meaning it fails to satisfy IC.}
    \item To address this, \modified{we propose a general methodology to design monotone greedy allocation mechanisms. We then apply this theory to produce a novel two-stage mechanism with carefully designed payments to ensure both IR and BB, as well as IC and scalability. }
	\item We empirically evaluate the performance of different mechanisms and show the advantages of our methods. Our monotone greedy mechanism satisfies all desirable properties with typically less than a 9\% social cost increase compared to the optimal allocation.
\end{itemize}
In what follows, Section~\ref{sec:model} introduces the ridesharing model and its mechanism design formulation. Section~\ref{sec:vcg} considers the VCG mechanism and its budget balanced version. Section~\ref{sec:greedy} studies the greedy mechanisms. Section~\ref{sec:exp} provides an empirical evaluation of the mechanisms and Section~\ref{sec:conclusions} concludes.  Due to space limitations, most proofs for the theoretical results are provided in \camera{the extended version of the paper}. 



\section{The Model}
\label{sec:model}
In this section, we first provide a formal model of the ridesharing problem followed by an example. 
We then formulate this as a mechanism design problem.

\subsection{The Ridesharing (RS) Problem}
\label{sec:pmodel}

We assume finite time steps  $\mc{T}=\{0,1,\ldots,T\}$, set of riders  $\mc{N}=\{1,2,\ldots,N\}$ and vehicles $\mc{K}=\{1,2,\ldots,K\}$.  Let $w \in \mathbb{N}_{+}$ denote the vehicle capacity (i.e., the maximum number of riders in a vehicle). The road network is represented by a directed graph $\mc{G}=(\mc{V},\mc{E})$, where edges $\mc{E}$ denote road segments, vertices $\mc{V}$ denote locations, \modified{and loops $(v,v) \in \mc{E}$ are used to model stationary vehicles}.  
\modified{The length} of an edge $e\in \mc{E}$ is given by $l_{e}$ and denotes the travel distance. For simplicity \modified{we assume $l_{(v,v')}=1$ if $v\not=v'$, and $l_{(v,v)}=0, \forall (v,v')\in \mc{E}$.\footnote[1]{The former ensures that a rider is always at a vertex at each time step. Furthermore, since $l_{e}$ is used to compute costs (discussed below), the latter ensures that stationary vehicles do not incur fuel costs. Longer roads can be represented by joining multiple edges.}}
Let $O_i^p,D_i^p \in \mc{V}$ denote the origin and the destination of rider $i\in \mc{N}$. Also, let $O_k^c$ denote the initial location of vehicle $k \in \mc{K}$, which does not have to be where riders are (in case of autonomous vehicles). We denote by $\mc{P}$  the set of all paths in $\mc{G}$. A \emph{path} is a vector of size $S$, denoted by $\langle v\rangle_S \equiv \langle v_1,\ldots,v_S\rangle$, where an element $v_s \in \mc{V}$ is the $s$-th location in the path and $(v_s,v_{s+1}) \in \mc{E}, \forall s < S$. Since paths do not include information of time, we define 
a \emph{route} as a tuple of vectors $\langle \langle v\rangle_S,\langle t\rangle_S\rangle$ where the path departs from location $v_s$ at time $t_s$. Route $\langle \langle v\rangle_S,\langle t\rangle_S\rangle$ satisfies $t_{s+1} = t_{s}+l_{(v_s,v_{s+1})}$ if $v_s \not =v_{s+1}$. If $v_s=v_{s+1}$, it means the rider or vehicle stays at $v_s$ for an amount of time given by $t_{s+1}-t_s$ (e.g., to wait for a vehicle or because they already arrived at the destination). The set of routes is denoted by $\mc{R}$. The route of rider $i$ and the route of vehicle $k$ are denoted by $r_{i}^p, r_{k}^c \in \mc{R}$, respectively. 
We let ${\rm edge}(r,t) \in \mc{E}$ be a function that returns the edge in route $r$ that the rider or vehicle is traversing at $t$. E.g., ${\rm edge}(\langle \langle A,B,C\rangle,\langle 0,1,2\rangle \rangle,1)=(B,C)$.  Similarly, ${\rm loc}(r,t) \in \mc{V}$ returns the location at $t$. E.g., ${\rm loc}(\langle \langle A,B,C\rangle,\langle 0,1,2\rangle \rangle,1)=B$.
Riders cannot move without vehicles but vehicles can move without riders if they are autonomous. If they are not autonomous, then they need at least one rider. 

We denote the \emph{assignment} of vehicles to riders with \modified{tensor} $B\in\mc{B}=\{0,1\}^{T\times N\times K}$. $B[t,i,k]=1$ is possible only if the path of rider $i$ is the same as vehicle $k$ at time $t$. 
The assignment may change over time, which means a rider can switch vehicles on the way\footnote[2]{In practice, this can take time, but we ignore this for simplicity. The model can  easily be extended by introducing waiting times for vehicles when riders switch \camera{at some designated safe areas}.}. We do not specifically distinguish between drivers and riders. Any rider in a vehicle can be a driver if the vehicle is not autonomous. Furthermore, riders can use a taxi as the outside option \modified{($B[t,i,k]=0,~\forall t\in\mc{T},k\in\mc{K}$)} if it is more beneficial than using ridesharing. We assume that the waiting time for a taxi is zero and the taxi takes the shortest path between the origin and destination of the rider. However, a rider cannot combine a taxi with ridesharing. 

An overall \emph{allocation} is denoted by
$\pi =\langle \vect{r}^p,\vect{r}^c,B \rangle \in \mc{R}^N\times\mc{R}^K\times\mc{B}$, where $\vect{r}^p=\langle r_{i}^p \rangle_{i\in\mc{N}}, \vect{r}^c=\langle r_{k}^c \rangle_{k \in \mc{K}}$ are the routes for riders and vehicles,  respectively. We will also use $r^p_i(\pi)$, $r^c_k(\pi)$ to denote a route for rider $i$ and vehicle $k$ given allocation $\pi$.  An allocation is feasible if, for each rider $i$, the route $r^p_i$ consist of a path between $O_i^p$ and $D_i^p$, and the vehicle capacity is not exceeded.  Formally, the set of \emph{feasible allocations} is defined by:

\begin{align}
\Pi=&\{\langle \vect{r}^p,\vect{r}^c,B \rangle  \in \mc{R}^N\times\mc{R}^K\times\mc{B} |  \nonumber\\
&\forall t \in \mc{T}, i\in \mc{N}, k\in\mc{K}, r_i^p, r_k^c \in \mc{R}: \nonumber\\
 &{\rm loc}(r^p_i,0)=O_i^p,~{\rm loc}(r^p_i,T)=D_i^p,~{\rm loc}(r^c_k,0)=O_k^c, \nonumber\\
 & B[t,i,k]=1 \Rightarrow {\rm edge}(r_k^c,t)={\rm edge}(r_i^p,t),\nonumber\\
 & \exists t' \in \mc{T}, \forall k' \in \mc{K}: (B[t',i,k']=0, l_{{\rm edge}(r_i^p,t')}>0) \nonumber\\
 &\qquad \Rightarrow B[t,i,k]=0, \nonumber\\
&\sum_{i'\in \mc{N}}B[t,i,k]\leq w \}.
\end{align}
The first line of constraints specifies origins and destinations of the routes. The second constraint is the definition of $B$. The third constraint means that a taxi cannot be combined with ridesharing (if riders traverse an edge using a taxi, they always use a taxi). The fourth is the capacity constraint.

If vehicles are not autonomous or have no dedicated driver, $\Pi$ needs an additional condition that any moving vehicle needs at least one rider: $l_{{\rm edge}(r_k^c,t)}>0 \Rightarrow \sum_{j\in \mc{N}}B[t,j,k]\modaaai{>} 0$.

Now that we have defined the allocation, we proceed with defining the costs which the system aims to minimise. These costs include travel and waiting time, fuel costs and taxi labour costs. Formally, the travel time of rider $i$ including waiting time is denoted by $T_i(\pi)=\min\{t_s|r_i^p=\langle \langle v\rangle_S,\langle t\rangle_S\rangle, v_s =D_i^p\}$\modified{, which is the earliest time to reach the destination}. Furthermore, let $T_i^0$ denote the travel time for rider $i$ using the shortest path (e.g., taxi).
The travel time of vehicle $k$ is denoted by 
$T_k^c(\pi)=\sum_{t \in \mc{T}}l_{{\rm edge}(r_k^c(\pi),t)}$.
Furthermore, let $\alpha,\beta \in \mathbb{R}$ denote the (additional) labour cost per time step for a taxi driver and fuel cost per distance for a vehicle, respectively. Let $\gamma_i \in \Gamma=[0,\gamma_{max}]$ represent the value of time, i.e., a subjective  cost rider $i$ is willing to pay to avoid a unit of travel time. Then the cost of rider $i$ is: $c_i(\pi)=c_i^0=(\alpha+\beta+\gamma_i)T_i^0$ when using a taxi, and $c_i(\pi)=\gamma_{i}T_i(\pi)$ when using ridesharing.   In addition, the fuel cost of the ridesharing system is calculated separately and is given by $c_F(\pi)=\sum_{k \in \mc{K}}\beta T_k^c(\pi)$. Given this, the {\it ridesharing (RS) problem} is to find a feasible allocation $\pi\in\Pi$ that minimizes the social cost including fuel cost $C_F(\mc{N},\pi)=\sum_{i \in \mc{N}}c_i(\pi)+c_F(\pi)$.

To illustrate the model, \modified{Figure \ref{fig:tradgreedy} shows an example with $N=2, K=2$ and $T=5$. The two figures show two possible allocations: in left figure, for example, $r_1^p=\langle \langle A,A,B\rangle,\langle 0,3,4\rangle \rangle$  and $ r_2^p=\langle \langle B,B,C\rangle,\langle 0,1,2\rangle \rangle$. Rider 2 is assigned to vehicle 1 at $t=1$, i.e., $B[1,2,1]=1$ while rider 1 is not assigned to vehicle 2  until $t=3$, i.e., $B[3,1,2]=1$. If $\gamma_1=2$ and $\gamma_2=1$, then the optimal allocation is the one on the right, with $T_i=T_k^c=3$ and social cost $C_F=15$.}

\subsection{Mechanism Design Formulation}
\label{sec:mmodel}
A mechanism takes reports from riders on their privately-known types as input, and computes the allocation and payment as output. In our model, we assume $O_i^p$ and $D_i^p$ are easy to track \modified{with GPS} and so $\gamma_i$ is the only private information, which is referred to as rider $i$'s \emph{type}. Let $\vect{\gamma}\in \Gamma^N$ be the type profile of all riders, $-i$ the set of riders except for $i$, and $\vect{\gamma}_{-i}$ the types of all riders except $i$.
For simplicity, we assume that all riders and vehicles are ready to depart at $t=0$, but this can be easily extended.  %
%
A \emph{mechanism} is then defined as a function $\mc{M}:\Gamma^N \to \Pi \times \mathbb{R}^N$, that  computes a feasible allocation $\pi \in \Pi$ from a given type profile and payment vector $\vect{x}(\pi)\in \mathbb{R}^N$, where $x_i(\pi)$ is the payment of rider $i$ to the system. We sometimes use  $\pi(\vect{\gamma})$ to emphasize the dependence on  $\vect{\gamma}$. If $i$ chooses a taxi, $x_i(\pi)=0$.  Sometimes we use short-hand notation $\pi(\gamma_i)=\pi(\gamma_i,\vect{\gamma}_{-i})$.
Given this, the utility of rider $i$ is defined as $u_i(\pi)=-c_i(\pi)-x_i(\pi)$. 

\camera{In practice, the ridesharing mechanism may work as in the following example. First, rider $i$ reports their type $\gamma_i$ using a slider in a suitable smartphone app, along with $O_i^p$ and $D_i^p$. Note that $\gamma_i$ is not necessarily truthful. Once the report profile $\vect{\gamma}$ is obtained, the mechanism computes $\pi(\gamma)$ and $\vect{x}(\pi)$, and reveals them to riders before travel. 
Finally, $x_i$ is charged after the travel.}
\camera{To this end}, the mechanism \camera{needs} to achieve the following properties:

    Budget balance (BB), which is formulated as:
    \begin{equation}
\sum_{i \in \mc{N}}x_i(\pi)\ge c_F(\pi).
\end{equation}

    Individual rationality (IR), which is formulated as:
\begin{equation}
u_i(\pi)\ge-c_i^0, ~\forall i \in  \mc{N}.
\label{eq:ir}
\end{equation}

Dominant-strategy incentive compatibility (DSIC), or truthfulness, which is formulated as:
\begin{equation}
 \left.
 \begin{array}{l}
u_i(\pi(\gamma_i,\vect{\gamma}_{-i}))\ge u_i(\pi(\gamma_i',\vect{\gamma}_{-i})), \\
~\forall \gamma_i' \in \Gamma,  ~\forall i \in  \mc{N}, \forall \gamma_{-i} \in \Gamma^{N-1}.
 \end{array}
 \right.
\end{equation}
 
 
\section{The VCG Mechanisms}
\label{sec:vcg}

\modified{Next, we will show that the well-known VCG mechanism \cite{Shoham2008} does not achieve BB and polynomial-time computation at the same time when applied to the RS problem.} This mechanism first computes an optimal allocation $\pi(\vect{\gamma})$ given (reported) type profile $\vect{\gamma}$ and then it computes the payment $\vect{x}(\pi)$ as the externality that each rider imposes on other riders. Specifically, for our RS problem, the objective is to minimize $C_F(\mc{N},\pi)$, and so the VCG payment is computed as $x_i(\pi)=C_F(-i,\pi(\vect{\gamma}))-C_F(-i,\pi(\vect{\gamma}_{-i}))$. As a result, the optimal allocation $\pi$ and the payment $\vect{x}$ guarantee DSIC and IR. 
However, we show that the RS problem is NP hard by reducing the travelling salesman problem. Also, VCG has the following issue.

\begin{minortheorem}
\modified{VCG does not satisfy BB in the RS problem.}
\end{minortheorem}
\begin{proof}
Suppose both riders 1 and 2 have the same origin $O$ and destination $D$, which are directly connected with an edge. VCG allocates the vehicle to both riders for sharing. Since the absence of either rider does not change the allocation and cost of the other rider, the VCG payment is zero for both riders. This means the fuel cost of the vehicle is not recovered.
\end{proof}

\modified{Despite this negative result, we show it is possible to modify VCG to satisfy BB by sacrificing the optimality, as follows.} Let $\mc{K}_i(\pi)=\{k \in \mc{K} |B[t,i,k]=1, \exists t \in \mc{T}\}$ denote the set of vehicles that rider $i$ uses. Then, the total fuel costs of all rides that a rider $i$ is involved in is $c_F^i(\pi)=\sum_{k \in \mc{K}_i}\beta T_k^c(\pi)$. We then define an \emph{imaginary} system cost by redundantly counting the fuel costs of vehicles as follows: $c_I(\pi)=\sum_{i \in \mc{N}}c_F^i(\pi)$. The total social cost including this imaginary cost is then given by $C_I(\mc{N},\pi)=\sum_{i \in \mc{N}}c_i(\pi)+c_F(\pi)+c_I(\pi)$. 
Given this, we have a budget balanced VCG (BVCG) 
by replacing the objective function of VCG with $C_I$. 
\begin{minortheorem}
\modified{Assuming autonomous vehicles, BVCG satisfies DSIC, IR and BB.}
\end{minortheorem}
\modified{Briefly, this is because the payment is computed based on the imaginary cost $c_I$, which is redundant and larger than $c_F$. Autonomous vehicles guarantee the feasibility of vehicle routing even when $i$ is absent in the payment computation. However, scalability is an issue as the standard VCG.}

\section{Greedy Mechanisms}
\label{sec:greedy}

\modified{To address the issues with VCG mechanisms,} we consider greedy mechanisms in this section.  
To this end, we first show that a naive greedy approach does not satisfy IC and we 
propose a general theory to design greedy DSIC mechanisms based on the concept of monotonicity. Then, we apply our theory to design a greedy DSIC ridesharing mechanism.

\subsection{Monotonicty and Violation with Naive Greedy}
\label{sec:naivegreedy}

\cite{Myerson1981,archer2001truthful} showed that, for single-valued domains, monotonicity of the allocation function $\pi(\vect{\gamma})$ is a necessary and sufficient condition for the existence of  payments $\vect{x}(\pi)$ satisfying DSIC. Using notation from our model, monotonicity is defined as follows:
\begin{equation}
    \gamma_i \leq \gamma_i' \Rightarrow V_i[\pi(\gamma_i)] \leq V_i[\pi(\gamma_i')], ~\forall \gamma_i, \gamma_i' \in \Gamma,
    \label{eq:mono}
\end{equation}
where, in general, $V_i:\Pi\to \mathbb{R}$ is the amount allocated, and referred to as \emph{work} in \cite{archer2001truthful} and, for our setting, $V_i=-T_i(\pi)$ corresponds to the negative travel time. 

Although monotonicity plays a key role in designing DSIC mechanism, finding appropriate monotone allocation functions $\pi$ is non-trivial. \modified{Here, we show that a naive approach fails to implement a DSIC mechanism. We define a naive greedy allocation as follows. First, all riders are assigned to a taxi by default. Then, for each iteration it chooses rider $i$ who 
minimizes the marginal social cost, given by 
$c_i-c_i^0+c_F$, without changing the allocation in riders allocated so far, and computes the best allocation for the current rider.}

\begin{figure}[t]
\centering
\includegraphics[width=1.\columnwidth]{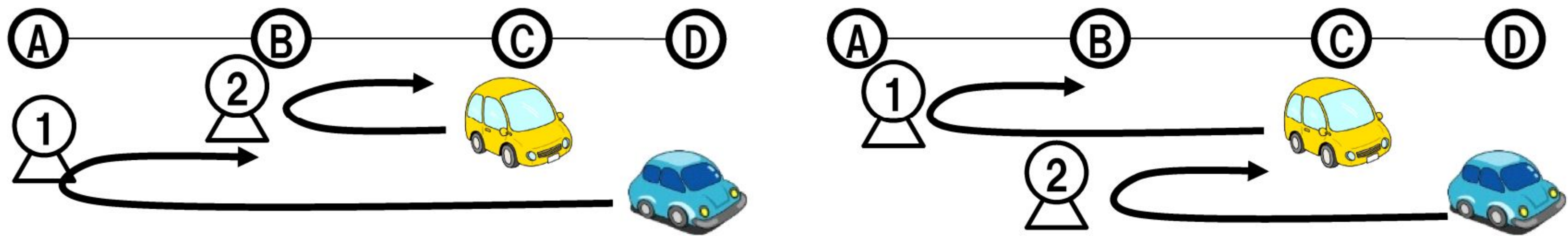}
\caption{\modified{Example showing how monotonicity is achieved with our mechanism and violated with a naive greedy mechanism: $D_1^p=B,D_2^p=C,O_1^c=C,O_2^c=D$. Assume $l_e=1, \forall e\in \mc{E}$ and expensive taxi $\alpha=10$ and $\beta=1$.}}
\label{fig:tradgreedy}
\end{figure}

\begin{minortheorem}
\modified{Mechanisms based on the naive greedy allocation do not satisfy monotonicity or DSIC.}
\label{thm:naivegreedy}
\end{minortheorem}

\begin{proof}
\modified{Figure \ref{fig:tradgreedy} also shows a counterexample whereby a naive greedy allocation fails to satisfy monotonicity. Let $\gamma_2=4$. If $\gamma_1=3$, the naive greedy allocates vehicle 1 to rider $2$ in the first iteration because the rider brings the lowest cost $\gamma_2 T_2-(\alpha+\beta+\gamma_2)T_2^0+c_F=-5$ (left figure). However, if $\gamma_1=1$, the cost of the rider decreases as $\gamma_1 T_1-(\alpha+\beta+\gamma_1)T_1^0+c_F=-6$ and rider $1$ is allocated first (right figure). Then, the lower $\gamma_1$ results in the earlier allocation and larger $V_1=-T_1$, which is the violation of the monotonicity (\ref{eq:mono}). This violation of monotonicity means mechanisms with the naive greedy allocation also violate DSIC, since this is a necessary condition for DSIC.}
\end{proof}

\subsection{Design of Monotone Greedy Allocation}
\label{sec:theory}

To \modified{address the issue above}, we propose a characterisation that is useful for designing monotone allocation functions for greedy mechanisms. \modified{As we saw, the naive greedy violates monotonicity by giving priority to the riders with lower $\gamma_i$.}
\modified{Instead, we allocate agents}
in descending order of their (reported) type $\gamma_i$.
In what follows, we use the index to denote the order in which agents are allocated.  I.e., $j<i$ implies $\gamma_j \geq \gamma_i$ and that $j$ is allocated before $i$. Let $\pi_{\leq i}$ denote the allocation up to and including $i$ (with remaining agents allocated arbitrarily or, e.g., in the context of ride sharing, allocated to a taxi). 
Also, let $\Pi_{\leq i}(\pi_{\leq i-1})$ or simply $\Pi_{\leq i}$ denote the set of possible allocations given $\pi_{\leq i-1}$, keeping agents $j<i$ fixed. 
Given this, we define a class of \emph{greedy allocations with decreasing choices (GADC)} as follows:
%
\begin{equation}
\left.
\begin{array}{l}
    \pi_{\leq i}=\argmin_{\pi \in \Pi_{\leq i}} J(\pi)
\end{array}
\right.
\label{eq:gars}
\end{equation}
subject to $\Pi_{\leq i} \subseteq \Pi_{\leq {j}}$ if $j<i$. Here,  $J:\Pi\to \mathbb{R}$ is an objective function. In words, GADC captures allocation functions which compute the optimal allocation for each $i$ individually given a set of possible allocations, \emph{and this set reduces as $i$ reports a smaller $\gamma_i$} (and so appears later in the sequence).

\begin{theorem}
    Let $\pi_{\leq i}(\gamma_a),\pi_{\leq i}(\gamma_b)$ denote $i$'s allocation when reporting $\gamma_a$,$\gamma_b$ respectively. A GADC satisfies monotonicity if and only if, $\forall \gamma_a, \gamma_b, i \in \mc{N}$:
    \begin{equation}
    J(\pi_{\leq i}(\gamma_a))\leq J(\pi_{\leq i}(\gamma_b)) 
    \Rightarrow V_i[\pi_{\leq i}(\gamma_a)]\geq V_i[\pi_{\leq i}(\gamma_b)]. 
    \label{eq:consist}
    \end{equation}
    \label{thm:nsmono}
\vspace{-0.2cm}
\end{theorem}
Intuitively, this means that, the more an agent contributes to decreasing the global objective (in our case, minimising the social cost), the better their allocation should be (in our case, the lower the travel time).  

\subsection{Monotonic Greedy Ridesharing Mechanism}
\label{sec:greedyrs}

We now apply Theorem~\ref{thm:nsmono} to our setting. First, we customize GADC from section \ref{sec:theory} for our ridesharing problem. In particular, $\Pi_{\leq i}(\pi)$ 
is redefined as follows:

\begin{align}
\Pi_{\leq i}(\langle \vect{r}^p,\vect{r}^c,B \rangle) & =\{ \langle \vect{r'}^p,\vect{r'}^c,B' \rangle \in \Pi | \nonumber\\
&\forall t \in \mc{T}, j<i,j'>i, k \in \mc{K}: \nonumber\\
& B'[t,j,k]=B[t,j,k], {r'}^p_j=r^p_j, 
\label{eq:shrink}\\
& B[t,j',k]=0 \}.
\label{eq:greedyfeasrs}
\end{align}

This means that, when computing $i$'s allocation, the allocation of the prior riders $j<i$ are not changed. With this definition, the assumption $\Pi_{\leq i}\subseteq \Pi_{\leq j}$ holds, because \modified{constraint (\ref{eq:shrink})} 
shrinks the set.
Then, a greedy allocation for ridesharing is obtained by applying equation (\ref{eq:gars}), with the objective function $J(\pi)=\sum_{i\in \mc{N}}c_{i}(\pi)$. \modified{In addition, to obtain monotonicity,  (\ref{eq:consist}) must be satisfied. To this end, we will show that we require an additional constraint whereby riders cannot be allocated to a taxi (i.e., their outside option). Although this constraint seems restrictive, we will relax this again later so that riders can use taxis if this is in their best interest (to satisfy IR). In the following, we refer to this setting as greedy allocation for ridesharing with no-taxi constraint or \mbox{GARS-N}. Formally:}
\begin{equation}
\left.
\begin{array}{l}
\pi_{\leq i}=\underset{\substack{\pi \in\Pi^i}}{\argmin}\underset{\substack{i'\in \mc{N}}}{\sum}c_{i'}(\pi) \\
s.t.~l_{{\rm edge}(r_i^p,t)}>0\Rightarrow \underset{\substack{k \in\mc{K}}}{\sum} 
B[t,i,k]=1, \forall t\in \mc{T}. \\
\end{array}
\right.
\label{eq:gars_notaxi}
\end{equation}

\begin{lemma}
GARS-N satisfies monotonicity.
\label{thm:gars_nmono}
\end{lemma}

\modified{Briefly, this is because minimizing $J$ coincides with minimizing $c_i$ for each rider $i$, since the constraints (\ref{eq:shrink}) and (\ref{eq:greedyfeasrs}) do not allow $i$ to change the allocations and costs of other riders. If riders do not use a taxi, $V_i=-T_i$ is proportional to $c_i$ and the equivalent condition of monotonicity in Theorem \ref{thm:gars_nmono} is satisfied. With the case in Figure \ref{fig:tradgreedy}, GARS-N allocates rider 1 first (right figure) when $\gamma_1=5$ and $\gamma_2=4$, and then $T_1=3$. Meanwhile, if $\gamma_1=1$, rider 2 is allocated first (left figure) and then $T_1=4$. Contrary to the naive greedy, the higher $\gamma_1$ results in the earlier allocation and larger $V_1$, and this satisfies monotonicity. Note that, without the no-taxi constraint, monotonicity can be violated as follows. If a taxi is cheap ($\alpha=1$), then rider 1 uses a taxi when $\gamma_1=1$, resulting in $T_1=1$, which violates monotonicity. Also, note that the no-taxi constraint does not help  naive greedy because nobody uses a taxi in the proof of Proposition \ref{thm:naivegreedy}.}

Next, we consider the computation of the payment.  Let $\hat{T}_i(\pi)=T_i(\pi)-T_i^0$ denote the normalized travel time. Following \cite{Myerson1981}, the unique payment satisfying DSIC is given by:
\begin{equation}
    x_i = x_i^0+\int_0^{\gamma_i} \gamma \frac{d}{d\gamma}\hat{T}_i(\pi(\gamma)) d\gamma
	\label{eq:integ}
\end{equation}
 where $x_i^0$ is a value which is independent of $i$'s report.  Solving the integral is challenging in general since it is necessary to recompute the allocations given all possible types of the rider whose payment we are computing. However, in our case, because $\hat{T}_i(\pi)$ only depends on the order of the allocation, it is sufficient to vary the position in the ordering. We refer to the GARS-N allocation with this payment as the GARS-N mechanism. Also, this mechanism can be computed in polynomial-time because the problem (\ref{eq:gars_notaxi}) includes the computation of a single rider only with the routes of former riders fixed, and is not NP-hard.

\begin{lemma}
The GARS-N mechanism is DSIC and can be computed in polynomial-time.
\label{thm:garsn_dsic}
\end{lemma}

To ensure BB, we compute the base payment $x_i^0$ based on report independent parameters $\vect{T}_0$, $\hat{\vect{T}}_{min}$ and $c_{Fub}$, where $\vect{T}_0$ is a vector of $T_i^0$. Furthermore,  $\hat{\vect{T}}_{min}$ is a vector of normalized minimal travel times, which is computed by solving the following allocation problem:
\begin{equation}
\left.
\begin{array}{l}
\pi=\underset{\substack{\pi' \in\Pi}}{\argmin}\underset{\substack{i\in\mc{N}}}{\sum}\hat{T}_i(\pi') \\
s.t.~l_{{\rm edge}(r_i^p,t)}>0\Rightarrow \underset{\substack{k \in\mc{K}}}{\sum} 
B[t,i,k]=1, \forall t\in \mc{T}, \forall i \in \mc{N}.\\
\end{array}
\right.
\nonumber
\end{equation}
Finally, $c_{Fub}$ is an upper bound on the fuel cost $c_F$ of \mbox{GARS-N}. We assume that this parameter is given (in the experimental section, we explain how this is computed in our experimental setting).
The base payment is computed by splitting $c_{Fub}$ among riders,
proportional to \modaaai{each agent $i$'s relative shortest travel time, $T^0_i - \hat{T}_{min,i} - \min_j (T^0_j - \hat{T}_{min,j})$.} 
Since this base payment covers $c_{Fub}$, GARS-N is BB.

\color{black}

\begin{lemma}
If $c_{Fub}$ is given, the GARS-N mechanism is budget balanced.
    \label{thm:gbb}
\end{lemma}

We now discuss IR. Since the no-taxi constraint restricts outside options, it can cause the violation of IR in some settings. To address this issue, we use the following lemma.

\begin{lemma}
The GARS-N mechanism satisfies IR if
\begin{equation}
    \hat{T}_i(\pi(0))\leq [(\alpha+\beta)T_i^0-x_i^0]/\gamma_{max}, ~\forall i \in \mc{N}.
\label{eq:irs}
\end{equation}
\label{thm:irs}
\end{lemma}

Briefly, (\ref{eq:irs}) shows the lower bound of the cost that satisfies IR.
We can now define the GARS-N mechanism with individual rationality (GARS-NIR) as follows. Before executing GARS-N, riders that violate condition (\ref{eq:irs}) are excluded from the mechanism, and are allocated to a taxi. This process is repeated until there is no violation. Then GARS-N is executed with the remaining riders. Since no riders violate (\ref{eq:irs}), this mechanism satisfies IR. 

\color{\cameracolor}
Algorithm \ref{alg:gnir} shows the pseudocode of GARS-NIR (using Python-like list operations). It filters out riders who violate IR condition (\ref{eq:irs}) in line \ref{alg:gnir.filter} and allocates them to a taxi (line \ref{alg:gnir.taxi}). Then, \textproc{GARS-N} in Algorithm \ref{alg:gn} computes the allocation for riders who satisfy (\ref{eq:irs}). First, it sorts riders in descending order of reported $\gamma_i$ (line \ref{alg:gn.sort}). Then it computes an allocation by \textproc{GreedyAlloc} (line \ref{alg:gn.alloc}). From line \ref{alg:gn.pay1} to \ref{alg:gn.pay0}, it computes the second term of $x_i$ in (\ref{eq:integ}), by changing the order of riders. In line \ref{alg:gn.pay1}, variables are initialized. For each rider $i$ (line \ref{alg:gn.for1}), it changes the order of the rider (line \ref{alg:gn.for2} to \ref{alg:gn.ord2}) and the report with a consistent value $\gamma_i'$ (line \ref{alg:gn.gamma}), and computes an allocation with the new order $\mc{N}'$ (line \ref{alg:gn.newalloc}). Thus, the integral term in (\ref{eq:integ}) is computed (line \ref{alg:gn.integ}), and the payment $x_i$ is given by adding the base payment $x_i^0$ (line \ref{alg:gn.pay0}).
The \textproc{GreedyAlloc} function (Algorithm \ref{alg:ga}) computes the allocation for each rider one by one by solving (\ref{eq:gars_notaxi}) in line \ref{alg:ga.solve}.



\begin{algorithm}[t]                 
	\caption{GARS-NIR}
	\label{alg:gnir}
	\begin{algorithmic}[1]
		\Procedure{$\pi=$GARS-NIR}{$\vect{\gamma},\mc{N}$}
		\State{$\mc{N}_{IR}=[]$}
		\For{$i \in \mc{N}$} \label{alg:gnir.iter}
		\If{$i$ satisfies (\ref{eq:irs})}
		\label{alg:gnir.filter}
		\State{$\mc{N}_{IR}.{\rm append}(i)$} 
		\Else
		\State{$B[t,i,k]=0,~\forall t \in \mc{T},\forall k\in\mc{K}$}
		\label{alg:gnir.taxi}
		\EndIf
		\EndFor
		\State{Return \textproc{GARS-N($\vect{\gamma},\mc{N}_{IR}$)}}
		\EndProcedure
	\end{algorithmic}
\end{algorithm}

\begin{algorithm}[t]                 
	\caption{GARS-N}
	\label{alg:gn}
	\begin{algorithmic}[1]
		\Procedure{$(\pi,\vect{x})=$GARS-N}{$\vect{\gamma},\mc{N}$}
		\State{$\hat{\mc{N}}=$Sorted $\mc{N}$ in descending order of $\gamma_i$} \label{alg:gn.sort}
		\State{$\pi=\textproc{GreedyAlloc}(\vect{\gamma},\hat{\mc{N}})$} \label{alg:gn.alloc}
		\State{$t_i=\hat{T}_i(\pi),\Delta x_i=0,  \forall i \in \mc{N}$} \label{alg:gn.pay1}
		\For{$i \in \mc{N}$} \label{alg:gn.for1}
		\For{$order_i \in [\hat{\mc{N}}.index(i),\ldots,N]$}
			\label{alg:gn.for2}
    		\State{$\mc{N}'=\hat{\mc{N}},\vect{\gamma}'=\vect{\gamma}$}
        	\State{$j=\hat{\mc{N}}[order_i]$}
			\State{$\mc{N}'.{\rm pop}(\hat{\mc{N}}.index(i))$}
			\State{$\mc{N}'.{\rm insert}(order_i,i)$}
			\label{alg:gn.ord2}
			\State{$\gamma_i'=\gamma_j$}
			\label{alg:gn.gamma}
    		\State{$\pi'=\textproc{GreedyAlloc}(\vect{\gamma}',\mc{N}')$} \label{alg:gn.newalloc}
    		\State{$\Delta x_i=\Delta x_i+(\hat{T}_i(\pi')-t_i)\gamma_j$} \label{alg:gn.integ}
    		\State{$t_i=\hat{T}_i(\pi')$}
		\EndFor
        \State{$x_i=\Delta x_i+x_i^0$}
        \label{alg:gn.pay0}
		\EndFor
		\EndProcedure
	\end{algorithmic}
\end{algorithm}

\begin{algorithm}[t]
	\caption{Greedy Allocation}
	\label{alg:ga}
	\begin{algorithmic}[1]
		\Procedure{$\pi=$GreedyAlloc}{$\vect{\gamma},\mc{N}$}
		\For{$i \in \mc{N}$} \label{alg:ga.iter}
		\State{Compute $\pi_{\leq i}$ by (\ref{eq:gars_notaxi})} \label{alg:ga.solve}
		\EndFor
		\State{Return $\pi_{\leq i}$}
		\EndProcedure
	\end{algorithmic}
\end{algorithm}

\color{black}

Now we have the following result.

\begin{theorem}
The GARS-NIR mechanism satisfies DSIC, IR, BB and can be computed in polynomial-time.
\label{thm:garsnir}
\end{theorem}


Although GARS-NIR is not optimal, also note that no mechanism can simultaneously achieve optimality, DSIC, IR and BB \cite{myerson1983efficient}. Hence obtaining three of these properties is the best one can achieve. \modified{It is also difficult to bound the efficiency of GARS-NIR, because we cannot use the standard bounding technique with naive greedy allocations \cite{nemhauser1978analysis}, for the reason we discussed in Proposition \ref{thm:naivegreedy}.} 
\modified{Determining theoretical efficiency bounds remains an open problem for future work.}

\color{black}

\section{Experimental Results}
\label{sec:exp}

Given that we do not have formal efficiency bounds, we evaluate the mechanisms empirically using simulations using both synthetic and real data. For the former, we generate 
random road networks of 4 degrees maximum and average results over 32 simulations for each setting. For the latter, we use the New York taxi data set \cite{nycdata}. We assume autonomous vehicles for all our experiments.
 %
We vary the taxi cost $\alpha$ and choose these loosely based on prices from London taxis. If we normalize the fuel cost to be $\beta=1$, then $\alpha=5$ is considered a moderate value and $\alpha=1$ is cheap. As for the value of time, we set $\gamma_i=\gamma_{max}*i/N$ to simulate heterogeneous riders. $c_{Fub}$ is estimated by sampling solutions of GARS-N with different allocation orders, and multiplying the maximum fuel cost among the solutions with a constant coefficient. \camera{Further details are provided in the extended version.}


In what follows we first compare the performance of all mechanisms on a small setting to enable comparison with non-scalable mechanisms such as BVCG. We then evaluate the scalability of the greedy mechanism to larger settings. Since GARS-NIR computes all possible ways for a rider to switch  between  vehicles, this turns out to be time consuming computationally. Also, in practice, such switches are often not desirable. Hence, for large settings, we modify \modaaai{the feasible set $\Pi$} of GARS-NIR to allow a rider to use only a single vehicle (sGARS-NIR). This version still satisfies the same properties as GARS-NIR.

\subsection{Small Setting}
\label{sec:smallexp}

\begin{table}[t]
\centering
\begin{tabular}{|l|l|l|l|l|l|}
\hline
Mechanism & OPT & DSIC & IR & BB & POLY \\
\hline
VCG & Y & Y & Y & N & Y \\ 
BVCG & N & Y & Y & Y & N \\
Hungarian & N & Y & Y & N & Y \\
GARS-NIR & N & Y & Y & Y & Y \\
\hline
\end{tabular}
\caption{Characteristics of mechanisms.}
\label{tbl:mechs}
\end{table}

\begin{figure}[t]
\centering
\includegraphics[width=1.\columnwidth]{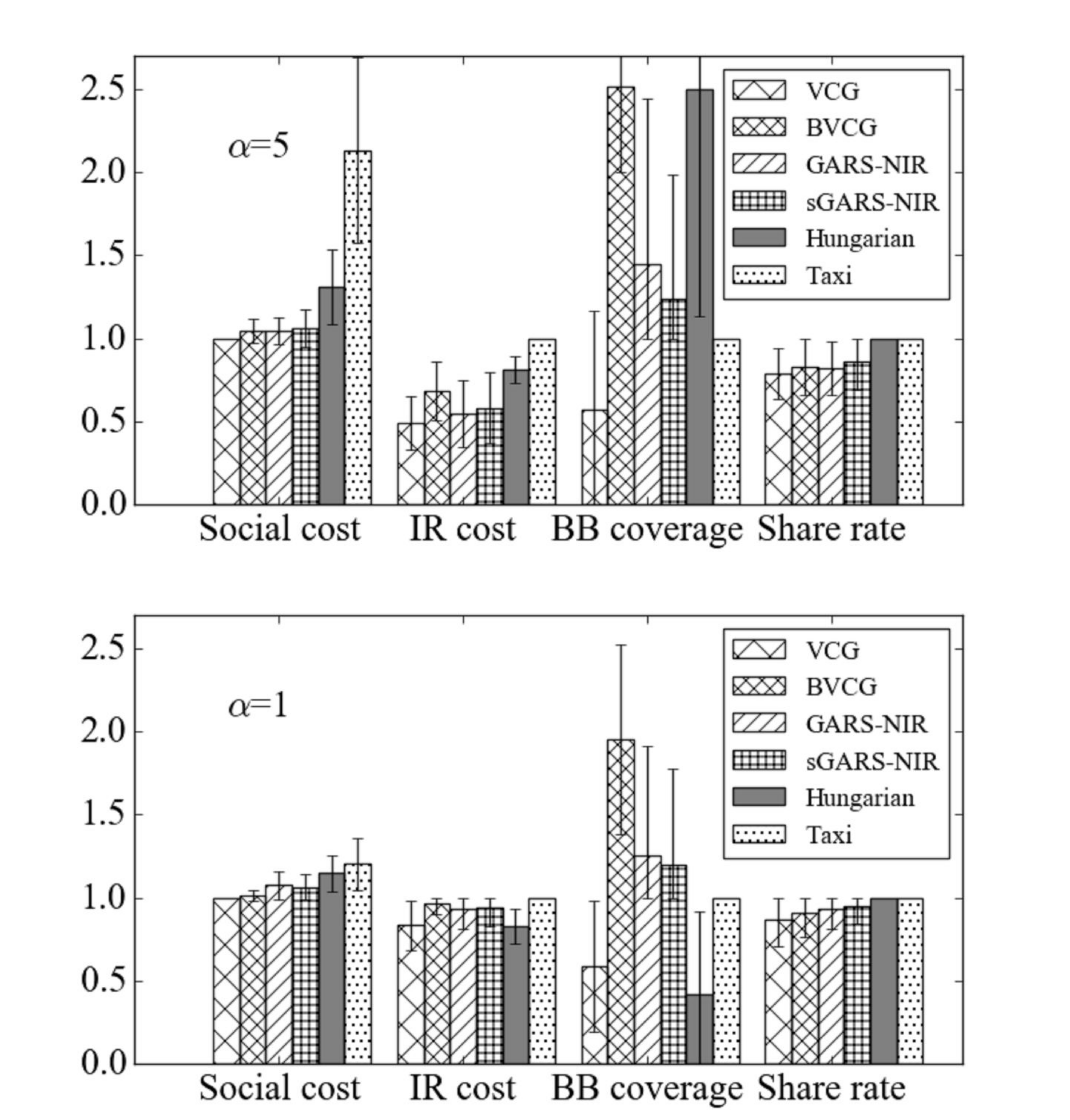}
\caption{Comparison of mechanisms with different taxi costs: \mbox{$\alpha=5$ (top),} $\alpha=1$ (bottom).}
\label{fig:results}
\end{figure}

We compare VCG, BVCG, GARS-NIR, the Hungarian mechanism and the outside option
\camera{(Table 1).}
The Hungarian mechanism computes an optimal 1-to-1 matching between riders and vehicles (including taxis) using the Hungarian algorithm, and then applies payments similar to VCG. The mechanism is also DSIC, BB and IR, but does not support multiple riders per vehicle. 
For BVCG, a brute force search is used to compute vehicle allocations. 
Figure \ref{fig:results} (top) shows the results for $N=3,K=2,T=4,V=4$ and $\alpha=5$ and the error bars show the 95\% confidence intervals. Here, \emph{social cost} is the \modaaai{ratio of $C_F(\mc{N},\pi)$ over the cost of the optimal allocation.} Furthermore, \emph{IR cost} refers to $-u_i(\pi)/c_i^0$. If this value is above 1, it means that IR is violated. \emph{BB coverage} shows 
$(\sum_{i \in \mc{N}}x_i(\pi))/ c_F(\pi)$.
A value less than 1 means BB is violated  \camera{(however, a value close to 1 is preferred, as long as it is over 1, meaning that riders are not over paying).}
\modaaai{Share rate shows $\sum_{k \in \mc{K}} T_k/\sum_{i \in \mc{N}} T_i$. A smaller value means ridesharing needs fewers vehicles to cover the trips of riders. }

Here, BVCG is budget balanced without a large increase in social cost (a 6.4\% increase). 
Meanwhile, VCG recovers only 52\% of the fuel costs on average in payments. However, a disadvantage of BVCG is its computational complexity when using the brute force search. Instead, \modaaai{GARS-NIR (and sGARS-NIR)} can be computed efficiently while satisfying DSIC, IR and BB, and still keeping social cost near optimal (an 8.6\% increase compared to optimal).  Also, it shows that the performance of sGARS-NIR is close to GARS-NIR.
\camera{Although the vehicle-switching behaviors of GARS-NIR are practical in suburban areas, it may not be so practical in urban areas. However, the result shows only $\frac{1}{6}$ of riders switch vehicles. Besides, it is possible to extend the model to introduce a waiting time for switching vehicles, and this can reduce the number of the cases. Also, it is easy to extend the model by mixing GARS-NIR with the version without switching (sGARS-NIR), and letting riders choose if they want to switch or not. 
}
In comparison, Figure \ref{fig:results} (bottom) shows the result when the outside option is cheap ($\alpha=1$). In this case, the IR cost comes closer to 1 for all mechanisms because of the incentive to opt out. BB coverage also decreases for some mechanisms because as more riders opt out, there is less opportunity to split the fuel cost. Meanwhile, the increase in social cost is not large in \modaaai{GARS-NIR} (4.5\%) because riders properly opt out \modaaai{if taxi is better}.
\modified{While 1.07 passengers ride on a vehicle on average in case of $\alpha=1$, it is 1.23 passengers in case of $\alpha=5$. This demonstrates that more expensive taxis lead to a higher usage of ridesharing, enabled by GARS-NIR satisfying IR of riders.}

\subsection{Scalability}
\label{sec:largeexp}

Next, we evaluate the scalability of sGARS-NIR. Since other mechanisms such as BVCG
are not scalable, we do not include them  here. Figure \ref{fig:totalelp} shows the runtime of the mechanism with different $N$ using parallel computation with 16 cores. Since the payment computation consists of independent computations of allocations with different orders of riders, it can be easily parallelized. For other parameters we use $K=N/2, T=15, V=10, \alpha=5, \beta=1$. sGARS-NIR takes less than 14 minutes when $N=40$. 
Figure \ref{fig:runtime}
shows that more than 99\% of the computation
time is spent constructing CPLEX LP objects. \modified{Hence, much of this can be precomputed and addressed with specalised implementation.}

\color{black}

\begin{figure}[t!]
\centering
\includegraphics[width=1.\columnwidth]{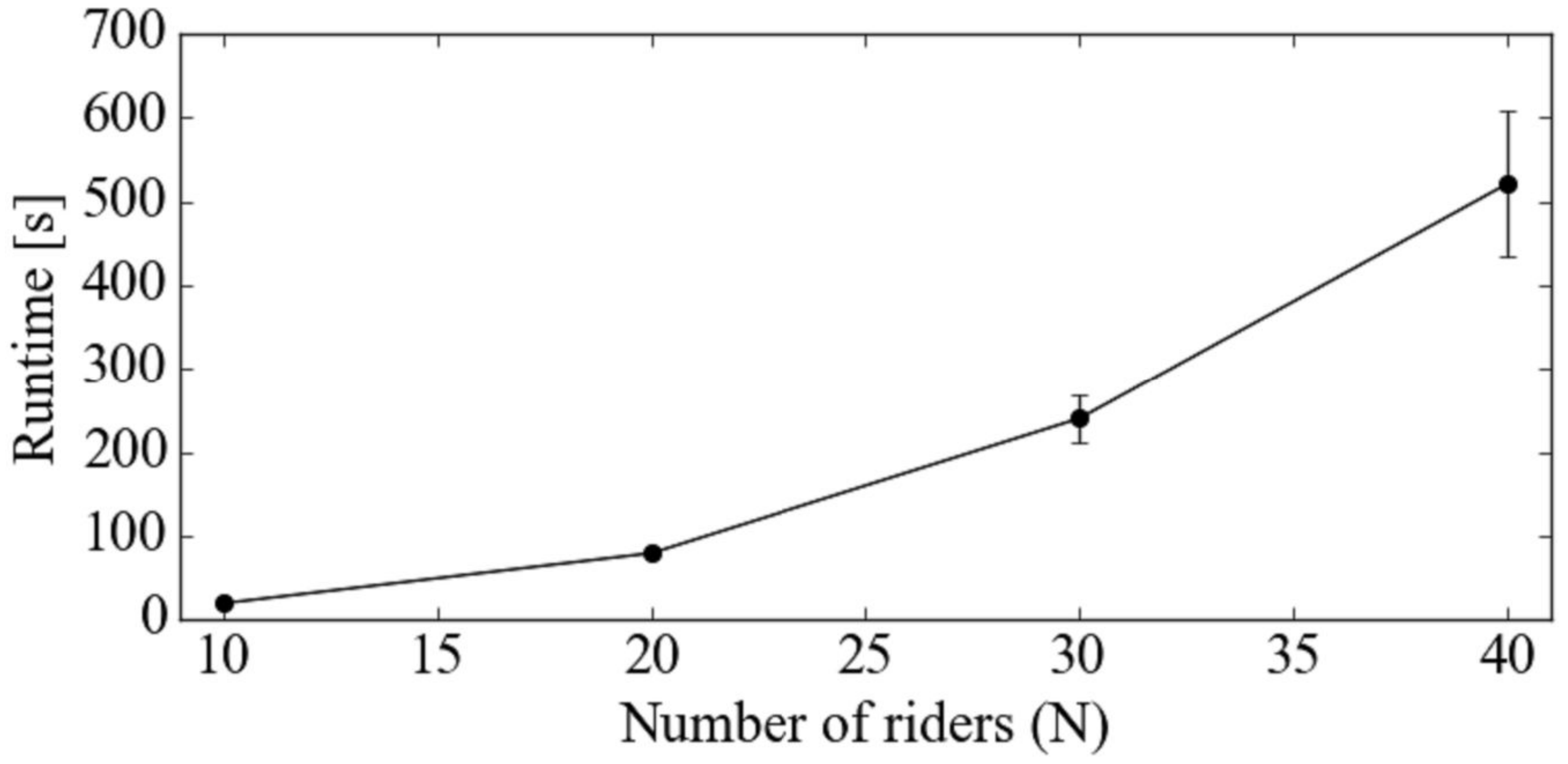}
\vspace{-0.2cm}%
\caption{Runtime of sGARS-NIR.}
\label{fig:totalelp}
\vspace{-0.2cm}
\end{figure}

\begin{figure}[t!]
\centering
\includegraphics[width=1.\columnwidth]{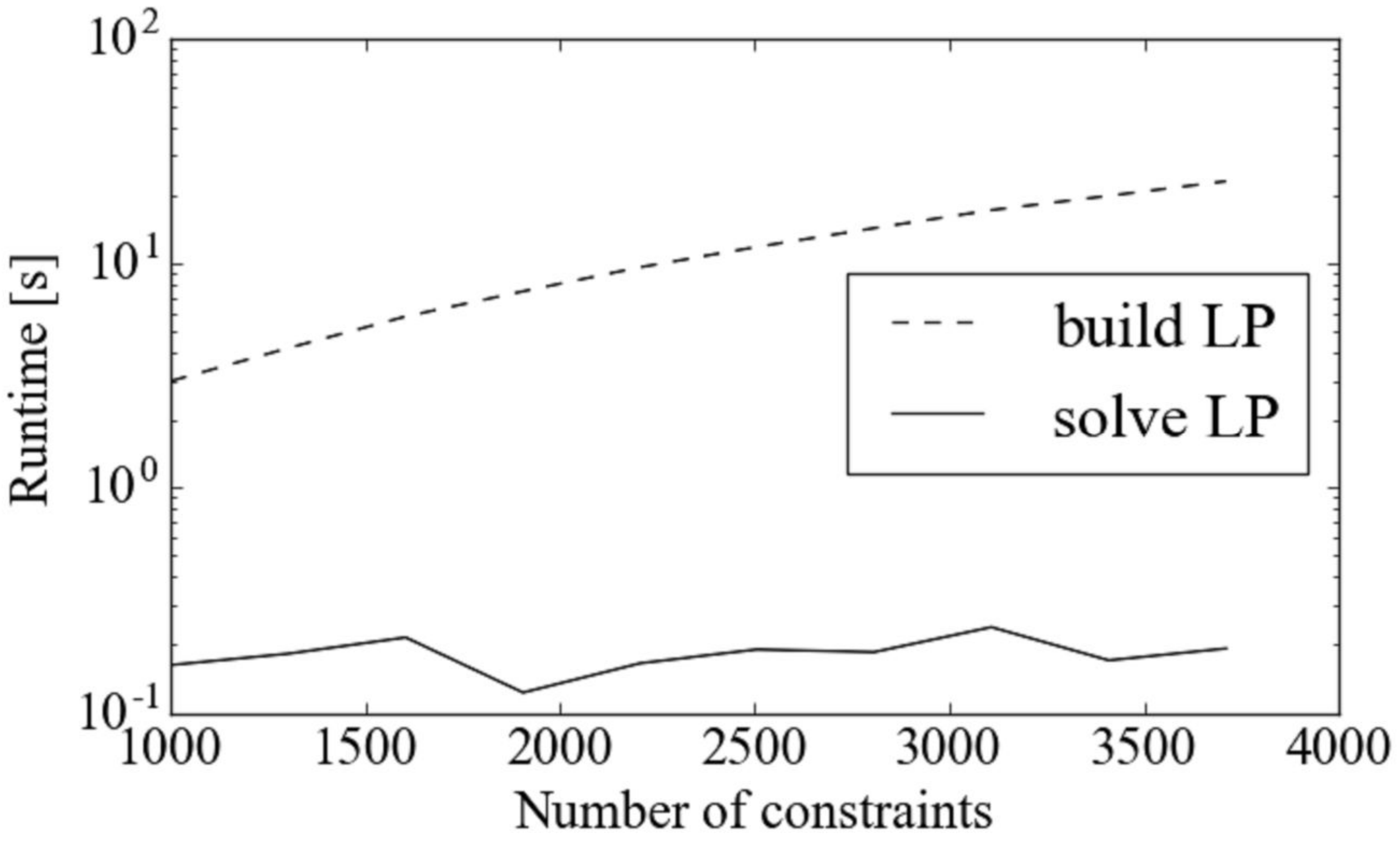}
\vspace{-0.2cm}
\caption{Figure showing a breakdown of the runtime. As can be seen, most of the runtime is spent constructing LP object, instead of solving it.}
\label{fig:runtime}
\vspace{-0.2cm}
\end{figure}
\begin{figure}[t!]
\centering
\includegraphics[width=1.\columnwidth]{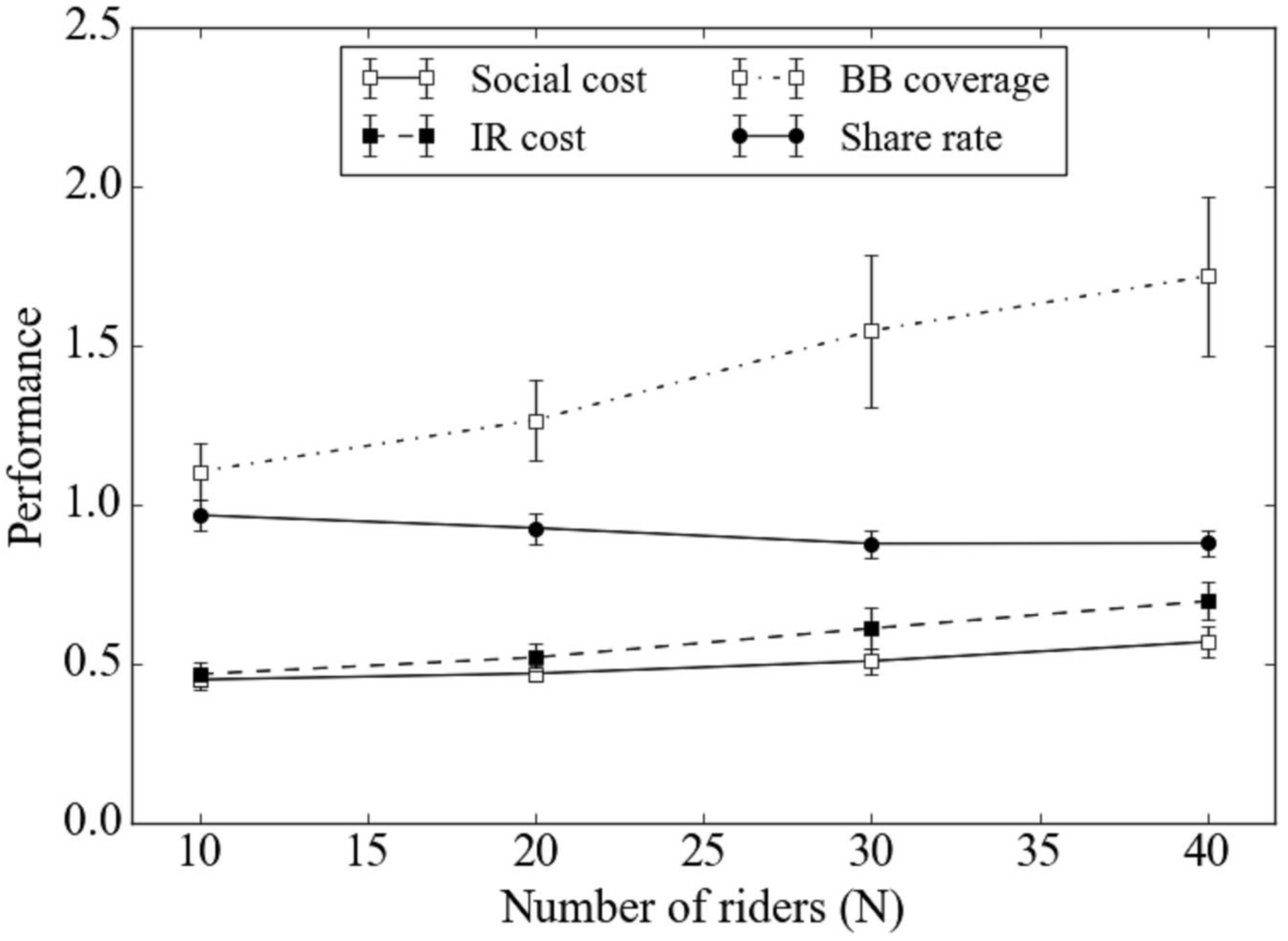}
\vspace{-0.4cm}
\caption{Performance as the number of riders (N) increases.}
\label{fig:nycsc}
\end{figure}

\subsection{Experiments with Real-world Data}
\label{sec:realexp}

In addition to using synthetic data, we also conduct experiments using the New York taxi data set \cite{nycdata}. We analyse the traffic demands in the central Manhattan area, which consists of 19 zones.
A single trip data point includes the origin zone, the destination zone, departure and arrival time. We randomly generate trip demands with the trip distribution from February 2019. All riders are assumed to be picked up and dropped off at the center of their zones. We set $K=20, T=15, V=19, \alpha=5, \beta=1$ and vary $N$. Figure \ref{fig:nycsc} shows the ratio of social cost compared to the taxi cost, IR cost, BB coverage and share rate of sGARS-NIR. 
\camera{The share rate is improved as $N$ increases since this increases the chance of sharing vehicles. However, with $K$ fixed, increasing $N$ causes a lack of vehicles and thus an increase in the IR cost and the social cost. This can be solved by providing an appropriate number of vehicles. The BB coverage shows over payment of riders with large $N$, but they still save compared to taking a taxi.}

\section{Conclusions}
\label{sec:conclusions}
We propose a flexible ridesharing model that includes autonomous vehicles.
The goal is to \modified{find allocation and payment mechanisms which are dominant-strategy incentive compatible,} budget balanced, individually rational, can be computed in polynomial time and achieve high social welfare. To this end, we \modified{show that both the VCG and naive greedy mechanisms do not meet these requirements.} We then present a  characterisation for monotone greedy allocations and  use this  to develop a novel greedy mechanism which satisfies all of the desirable properties. In addition, through numerical simulations we show that \modified{our} monotone greedy mechanism achieves close to optimal allocations. 

Future directions include \modaaai{envy-free mechanisms and} improving the implementation of the algorithm so it scales to even larger settings, as well as providing theoretical performance bounds for our  greedy approach. 

\section*{Ethical Impact}

Our work has positive ethical impacts on society. First, our mechanism provides a fair ridesharing service in the sense that people get compensated for their detour. Second, the incentive compatibility of our mechanism maximizes the utility of individuals. Third, our ridesharing mechanism can benefit society and the environment by meeting more traffic demand with fewer vehicles than existing taxi services or privately-owned vehicles. A possible negative impact is that people might perceive the mechanism as unfair because the payments for a given service may differ depending on the other participants. We could address this by looking into envy-free mechanisms in  future work.
\bibliographystyle{named}
\bibliography{main}

\if\arxiv1
\setcounter{secnumdepth}{0} 

%
\title{Scalable, Truthful and Budget Balanced Ridesharing Mechanisms \\ (supplemental material)}

\author{Paper ID 8002}

\begin{appendices}

\section{Appendix A: NP-hardness of the Ridesharing Problem}

\label{sec:lp}
\color{\newcolor}
We show the following negative result on the computational complexity of the RS problem.

\setcounter{minortheorem}{4}
\begin{minortheorem}
    Finding an optimal solution of the RS problem is NP-hard.
	\label{thm:NPh}
\end{minortheorem}

\begin{proof}
Consider the NP-hard Travelling Salesman Problem (TSP), i.e., the problem to find the shortest cycle in a graph that visits all vertices. Any TSP can be reduced to a special case of the ridesharing problem where: 
\begin{itemize}
    \item The road network, $\mc{G}=(\mc{V},\mc{E})$,  is the same network as in the TSP. There is one rider on every vertex. 
    \item There is only one vehicle $k$, and its capacity is $w=N$. It can be placed on any vertex of the road network.
    \item The taxi cost is $\alpha>\sum_{e\in \mc{E}}1$, which is too expensive for riders to use.
    \item The vehicle is autonomous.
    \item Riders' destination is equal to the origin of the vehicle as $D_i^p=O_k^c,~\forall i\in\mc{N}$.
    \item $\gamma_i=0,~\forall i\in\mc{N}$,  $\beta=1$ and $T = |\mc{E}|$.
\end{itemize}

As the only cost incurred in this version of the RS problem is the fuel cost, the vehicle will choose the shortest cycle in the graph to pick up all riders. This is the same solution as the original TSP problem. Since any TSP can be reduced to an RS problem in polynomial time, the RS problem is also NP-hard.
\end{proof}

\section{Appendix B: Complete Proofs of Theorems}
\label{sec:proof}

\subsection{Proof of Proposition 2}

\color{\newcolor}

\begin{proof}
The BVCG mechanism can be seen as a VCG mechanism with an additional player 0 added. In more detail, consider an allocation problem among riders (players $1,\ldots,N$) and a ridesharing system (player $0$) who does not move and has no need for cars. The cost function of the system (i.e., player 0) is $c_F(\pi)+c_I(\pi)$. We then use the VCG mechanism for this problem. Since BVCG is an instance of VCG, DSIC and IR of all riders are trivially satisfied. Note that we are not concerned about DSIC and IR of the system player and can ignore its payment. A similar approach was used by \cite{Jehiel2006} for a different setting. 


\color{black}

\modified{Now we prove the BB property by assuming autonomous vehicles.}
We first obtain the following lemma, which shows the relationship between the fuel cost and the imaginary cost for a given allocation.

\setcounter{lemma}{4}
\begin{lemma}
	$c_F(\vect{\pi}) \leq c_I(\vect{\pi}) \leq N*c_F(\vect{\pi})$
	\label{thm:imaginary}
\end{lemma}

Intuitively this is because, if riders use the same vehicle, everyone's individual imaginary cost can be equal to $c_F$ in the worst case.

Moreover, the following lemma describes the change in social cost when a rider is removed from the mechanism.

\begin{lemma}
    Assuming autonomous vehicles, \\
	$C_I(-i,\vect{\pi})-c_I^i(\vect{\pi}) \geq C_I(-i,\vect{\pi}_{-i})$.
	\label{thm:absence}
\end{lemma}

\begin{proof}[Proof of Lemma \ref{thm:absence}]
Let $\vect{\pi}\setminus i$ denote an allocation for the set of riders $-i$, which is completely the same as $\vect{\pi}$, with the same set of vehicles and the same routes. The only difference is the absence of rider $i$. Such a route is still feasible because of the autonomous vehicles assumption (where vehicles can move without riders). 

Because of the absence of $i$ in $\vect{\pi}\setminus i$, the system social cost loses the cost related to $i$, which is,

\begin{equation}
 \left.
 \begin{array}{l}
 C_I(-i,\vect{\pi}\setminus i)\\
 =C_I(\mc{N},\vect{\pi})-c_i(\vect{\pi})-c_I^i(\vect{\pi}) \\
 =C_I(-i,\vect{\pi})-c_I^i(\vect{\pi}).
 \end{array}
 \right.
\end{equation}
 
Although $\vect{\pi}\setminus i$ is feasible, it is not necessarily optimal. Due to the optimality of the allocation for $-i$, we have that:
 
\begin{equation}
C_I(-i,\vect{\pi}\setminus i) \geq C_I(-i,\vect{\pi}_{-i})
\end{equation}
 
Then the lemma follows.
\end{proof}

To apply Lemma \ref{thm:absence}, we start from the following term
\begin{equation}
x_i-c_I^i(\vect{\pi})= C_I(-i,\vect{\pi})-C_I(-i,\vect{\pi}_{-i})-c_I^i(\vect{\pi}).
\end{equation}

By Lemma \ref{thm:absence}, this is larger than 0 and then $x_i\geq c_I^i(\vect{\pi})$.

Then, by summing up for all riders, $\sum_{i \in \mc{N}}x_i\geq \sum_{i \in \mc{N}}c_I^i(\vect{\pi})=c_I(\vect{\pi})$. Then by Lemma \ref{thm:imaginary}, $\sum_{i \in \mc{N}}x_i\geq c_F(\vect{\pi})$. \modified{This completes the proof of the proposition.}
\end{proof}

\subsection{Proof of Theorem 1}

\color{\newcolor}
\begin{proof}
First we prove it is a sufficient condition. Assume  \extref{7}. If $\gamma_i \leq \gamma_k \leq \gamma_i'$, reporting $\gamma_i$ can only delay the allocation of $i$ compared to reporting $\gamma_i'$. With slight abuse of notation, let $\Pi_{\leq i}(\gamma)$ denote the set of possible allocations when reporting $\gamma$.  Then, $\Pi_{\leq i}(\gamma_i) \subseteq \Pi_{\leq k}(\gamma_k) \subseteq \Pi_{\leq i}(\gamma_i')$.
Then,
$J(\pi_{\leq i}(\gamma_i))\geq J(\pi_{\leq i}(\gamma_i'))$.
Then, because of \extref{7}, 
$V_i[\pi_{\leq i}(\gamma_i)]\leq V_i[\pi_{\leq i}(\gamma_i')]$, and then $V_i[\pi(\gamma_i)]\leq V_i[\pi(\gamma_i')]$.
This proves it is a sufficient condition.

Next, we prove it is also a necessary condition. Assume  \extref{5}.
We will prove this by contradiction, first assuming the negation of \extref{7}, which is
\begin{equation}
\left.
\begin{array}{l}
    J(\pi_{\leq i}(\gamma_a))\leq J(\pi_{\leq i}(\gamma_b))\wedge \\
    V_i[\pi_{\leq i}(\gamma_a)]< V_i[\pi_{\leq i}(\gamma_b)],
    \exists \gamma_a, \gamma_b, i \in \mc{N}.
    \label{eq:neg}
\end{array}
\right.
\end{equation}

The first inequality means $\gamma_b \leq \gamma_a$, as inferred above.
Meanwhile, $\gamma_b \leq \gamma_a$ and monotonicity (\extref{5}) give
    $V_i[\pi_{\leq i}(\gamma_a)]\geq V_i[\pi_{\leq i}(\gamma_b)]$,
which contradicts the second inequality in Equation \ref{eq:neg}.
This proves the necessary condition.
\end{proof}

\subsection{Proof of Lemma 1}

\begin{proof}[Proof of Lemma 1]
    Let $\gamma_a, \gamma_b$ be $i$'s reports and $\gamma_j > \gamma_a \geq \gamma_b$. Since the GARS-N is a GADC, then
    
    \begin{equation}
    J(\pi_{\leq i}(\gamma_a))\leq J(\pi_{\leq i}(\gamma_b)).
    \end{equation}
    
    Meanwhile, the constraints of 
\extref{8}
    does not allow $i$ to change the allocations and costs of other riders. Then,
    
    \begin{equation}
    \left.
    \begin{array}{l}
    \pi_{\leq i}=\underset{\substack{\pi \in\Pi_{\leq i}}}{\argmin}J(\pi)=\underset{\substack{\pi \in\Pi_{\leq i}}}{\argmin}~c_i(\pi).
    \end{array}
    \right.
    \end{equation}
    And then,
    $c_i(\pi_{\leq i}(\gamma_a)) \leq c_i(\pi_{\leq i}(\gamma_b))$. Since the riders later than $i$ cannot change $c_j(\pi_{\leq i})$ either, $c_i(\pi(\gamma_a)) \leq c_i(\pi(\gamma_b))$. Because of the no taxi constraint, $i$ uses ridesharing and then $T_i$ is proportional to $c_i$. Then 
    \begin{equation}
    T_i(\pi(\gamma_a)) \leq T_i(\pi(\gamma_b)).
    \label{eq:gars_alloc}
    \end{equation}

    Since the amount of the  allocation is the negative travel time as $V_i=-T_i$, we have proved 
    \extref{7}. This is equivalent to monotonicity according to Theorem 1.
    Note that, if $i$ could take a taxi to reduce $T_i$, then (\ref{eq:gars_alloc}) could be violated since a taxi follows the shortest path.
\end{proof}

\color{black}

\subsection{Proof of Lemma 2}
\begin{proof}
First, we prove that the payment scheme 
is a necessary condition of incentive compatibility. From the definition of incentive compatibility,
\begin{equation}
 \left.
 \begin{array}{ll}
&-\hat{c}_i(\gamma_i,\gamma_{-i})-x_i(\gamma_i,\gamma_{-i}) \geq -\hat{c}_i(\gamma_i',\gamma_{-i})-x_i(\gamma_i',\gamma_{-i}) \\
\therefore &(\alpha+\beta)T_i^0-\hat{T}_i(\gamma_i,\gamma_{-i})\gamma_i-x_i(\gamma_i,\gamma_{-i}) \\
&\geq (\alpha+\beta)T_i^0-\hat{T}_i(\gamma_i',\gamma_{-i})\gamma_i-x_i(\gamma_i',\gamma_{-i}) \\
\therefore& [\hat{T}_i(\gamma_i',\gamma_{-i})-\hat{T}_i(\gamma_i,\gamma_{-i})]\gamma_i \geq x_i(\gamma_i,\gamma_{-i})-x_i(\gamma_i',\gamma_{-i})
 \end{array}
 \right.
\end{equation}

Let $0<y<z$. Then,
$[\hat{T}_i(z)-\hat{T}_i(y)]y \geq x_i(y)-x_i(z)$ 
and also,
$[\hat{T}_i(y)-\hat{T}_i(z)]z \geq x_i(z)-x_i(y)$. 
Then,
$[\hat{T}_i(z)-\hat{T}_i(y)]y \geq x_i(y)-x_i(z) \geq [\hat{T}_i(z)-\hat{T}_i(y)]z$. 
Then, when $y\to z$ at a jumping point $z$,
$-\Delta\hat{T}_i(z)z=\Delta x_i(z)$.
Then, the incentive compatible payment must be

\begin{equation}
x_i(\gamma_i,\gamma_{-i})=-\underset{\substack{j \in {\rm jumps} \\ z_j \leq \gamma_i}}{\sum}z_j\Delta \hat{T}_i(z_j,\gamma_{-i}).
\label{eq:monopay}
\end{equation}

The payment 
satisfies this property. Next, we prove the payment is also a sufficient condition. This payment results in incentive compatibility when the allocation is monotone as in the Figure  \ref{fig:gic}. In the figure, $\gamma_i$ is the true type of $i$ and $\hat{\gamma}_i$ is reported value. Then the cost of $i$ is the area shown in gray and the payment is the area in black. Then the utility is maximized when $i$ reports its type honestly.
\end{proof}

\begin{figure}[t]
\centering
\includegraphics[width=1.0\columnwidth]{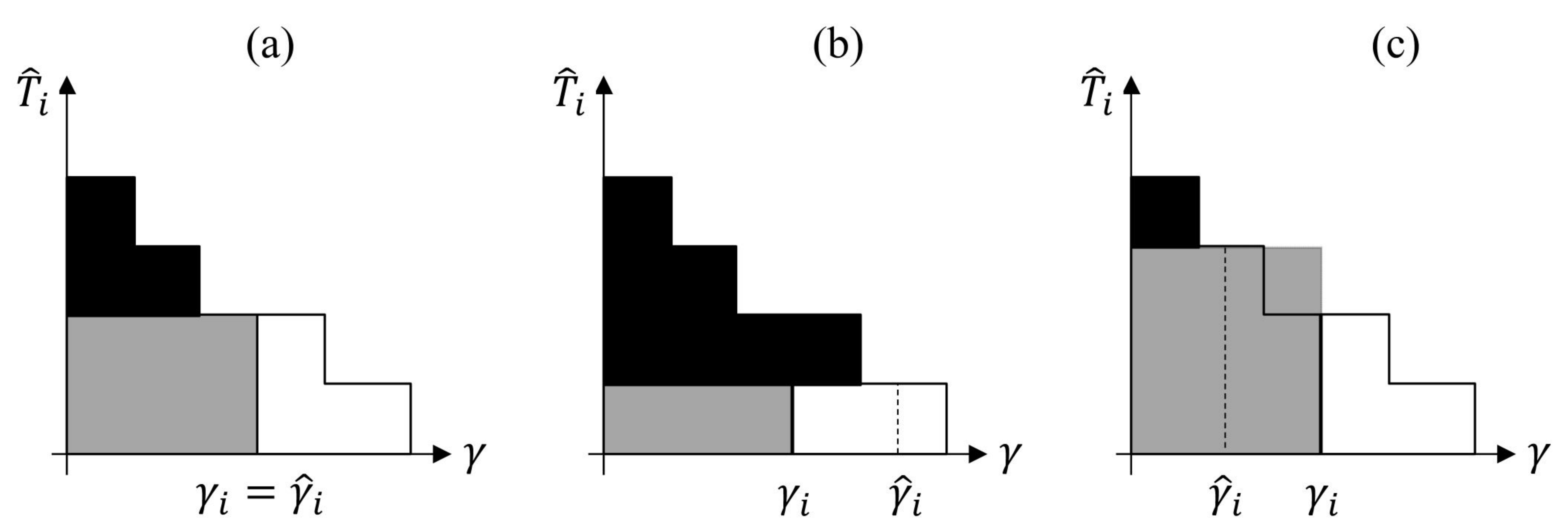}
\caption{Incentive compatibility with monotone allocation \\
(a) true bidding, (b) over bidding and (c) under bidding}
\label{fig:gic}
\end{figure}

\subsection{Proof of Lemma 3}
\begin{proof}
Since the base payment $x_i^0$ is computed by splitting $c_{Fub}$ and $x_i\geq x_i^0$,
\begin{equation}
\underset{\substack{i \in \mc{N}}}{\sum}x_i \geq \underset{\substack{i \in \mc{N}}}{\sum}x_i^0 = c_{Fub} \geq c_{F}.
\label{eq:gbbp}
\end{equation}
\end{proof}

\subsection{Proof of Lemma 4}
\begin{proof}
\begin{equation}
\left.
\begin{array}{l}
[(\alpha+\beta)T_i^0-x_i^0]/\gamma_{max}\geq \hat{T}_i(\vect{\pi}(0))\\
\therefore (\alpha+\beta)T_i^0\geq \gamma_{max}\hat{T}_i(\vect{\pi}(0))+x_i^0\geq \gamma_i\hat{T}_i(\vect{\pi}(0))+x_i^0\\
\therefore (\alpha+\beta)T_i^0-\gamma_i\hat{T}_i(\vect{\pi})\geq  \gamma_i[\hat{T}_i(\vect{\pi}(0))-\hat{T}_i(\vect{\pi})]+x_i^0 \geq x_i(\vect{\pi})\\
\therefore (\alpha+\beta+\gamma_i)T_i^0\geq \gamma_iT_i(\vect{\pi})+x_i(\vect{\pi})\\
\therefore c_i^0\geq c_i(\vect{\pi})+x_i(\vect{\pi})\\
\end{array}
\right.
\end{equation}
\end{proof}

\section{Appendix C: Details of Experiments}
\label{sec:expdetail}

\begin{itemize}
    \item We call numpy.random.seed(1) at the beginning of the main code. 
    \item We use CPLEX 12.7.1, Python 2.7.16, Red Hat Enterprise Linux Server release 6.10 and Intel Xeon CPU E5-2670 (2.60 GHz) with 8 cores, 132 GB memory to run the experiments.
    \item We estimate a realistic value for the cost parameter ${\bf \alpha}$ (wage of taxi drivers), with data from London taxis \cite{londonspeed,londonwage}. Specifically, fuel price: 124 [p/liter], wage of driver: 836 [GBP/week], working hour of driver: 50 [h/week], taxi speed: 30 [km/h] and mileage: 11.7 [km/liter]. Then, the estimated wage of taxi drivers is
    ${\bf 836/50/30*100=55.7}$ [p/km] and fuel cost is ${\bf 124/11.7=10.5}$ [p/km]. If we normalize the fuel cost as ${\bf \beta=1}$, then the normalized wage is ${\bf 55.7/10.5=5.3}$. Thus, we use ${\bf \alpha=5}$ as a moderate value.
\end{itemize}

\color{\cameracolor}
\begin{algorithm}                      
	\caption{Payment Parameter}
	\label{alg:pp}
	\begin{algorithmic}[1]
		\Procedure{$c_{Fub}=$PaymentParam}{}
	    \State{$c_{Fub}=0$}
		\For{$i \in \mc{N}$} \label{alg:pp.dev1}
		\For{$order_i \in [0,\ldots,N]$}
		\For{$j \in -i$}
		\For{$order_j \in [0,\ldots,N]$}
    	\State{$\mc{N}'=\mc{N},\vect{\gamma}'=\vect{\gamma}$}
    	\State{$k=\mc{N}[order_i]$}
		\State{$\mc{N}'.{\rm pop}(\mc{N}.index(i))$}
		\State{$\mc{N}'.{\rm insert}(order_i,i)$}
		\label{alg:pp.movei}
		\State{$\gamma_i'=\gamma_k$}
    	\State{$k=\mc{N}[order_j]$}
		\State{$\mc{N}'.{\rm pop}(\mc{N}.index(j))$}
		\State{$\mc{N}'.{\rm insert}(order_j,j)$}
		\label{alg:pp.movej}
		\State{$\gamma_j'=\gamma_k$}
		\label{alg:pp.dev2}

    	\State{$\pi=\textproc{GreedyAlloc}(\vect{\gamma}',\mc{N}')$} \label{alg:pp.alloc}
    		\If{$c_{Fub} < c_F(\pi)$}
    		\State{$c_{Fub} = c_F(\pi)$} \label{alg:pp.cfmax}
    		\EndIf
		\EndFor
		\EndFor
		\EndFor
		\EndFor
		\EndProcedure
	\end{algorithmic}
\end{algorithm}

\begin{algorithm}                      
	\caption{Base Payment}
	\label{alg:bp}
	\begin{algorithmic}[1]
		\Procedure{$\vect{x}^0=$BasePayment}{$\hat{\vect{T}}_{min},c_{Fub}$}
		\State{Compute loss: $\Delta \vect{T}=\vect{T}^0-\hat{\vect{T}}_{min}$} \label{alg:bp.t}
		\State{Normalization: $\Delta \vect{T}=\Delta \vect{T}-\min(\Delta \vect{T})$} \label{alg:bp.nt}
		\If{$\sum(\Delta \vect{T})==0$}
		    \State{$x_i^0=c_{Fub}/N, ~\forall i \in \mc{N}$}
		\Else
		    \State{$k=c_{Fub}/\sum(\Delta \vect{T})$} \label{alg:bp.coef}
		    \State{$x_i^0=k\Delta T_i,~\forall i \in \mc{N}$} \label{alg:bp.prop}
		\EndIf
		\EndProcedure
	\end{algorithmic}
\end{algorithm}

The \textproc{PaymentParam} function (Algorithm \ref{alg:pp}) shows a heuristics to estimate the upper bound of $c_F$.
It searches pairs of deviations of every two riders (line \ref{alg:pp.dev1}-\ref{alg:pp.dev2}) and computes allocations (line \ref{alg:pp.alloc}) to find the upper bound of the fuel cost $c_{Fub}$ (line \ref{alg:pp.cfmax}). 
The function \textproc{BasePayment} (Algorithm \ref{alg:bp}) shows how to compute the base payment $\vect{x}^0$.
Riders with longer trips and less waiting time pay more. For this purpose, it computes the normalized travel time (line \ref{alg:bp.t},\ref{alg:bp.nt}) and makes the base payment proportional to it (line \ref{alg:bp.coef},\ref{alg:bp.prop}).

\end{appendices}

\fi





\end{document}